\documentclass[3p,authoryear,12pt]{elsarticle}
\usepackage{graphicx}
\usepackage{amsmath}
\usepackage{amssymb}
\usepackage{amsthm}
\usepackage{natbib}
\usepackage{breqn}
\usepackage{longtable}
\usepackage{tabularx}
\usepackage{booktabs}
\usepackage{caption}
\usepackage{lscape}
\usepackage{hyperref}
\newcommand{\diag}{\mathop{\mathrm{diag}}}
\newcommand{\cov}{\mathop{\mathrm{cov}}}
\newcommand{\var}{\mathop{\mathrm{var}}}

\newcommand{\sgn}{\mathop{\mathrm{sgn}}}
\newtheorem{theorem}{Theorem}[section]

\newtheorem{lemma}{Lemma}[section]

\journal{Computational Statistics and Data Analysis}

\begin{document}

\begin{frontmatter}



\title{Weighted quantile regression for longitudinal data}


\author{Xiaoming Lu, Zhaozhi Fan}

\address{Department of Mathematics and Statistics, Memorial University of Newfoundland, St.John's, Newfoundland, Canada}

\begin{abstract}
Quantile regression is a powerful statistical methodology that complements the classical linear regression by examining how covariates influence the location, scale, and shape of the entire response distribution and offering a global view of the statistical landscape. In this paper we propose a new quantile regression model for longitudinal data. The proposed approach incorporates the correlation structure between repeated measures to enhance the efficiency of the inference. In order to use the Newton-Raphson iteration method to obtain convergent estimates, the estimating functions are redefined as smoothed functions which are differentiable with respect to regression parameters. Our proposed method for quantile regression provides consistent estimates with asymptotically normal distributions. Simulation studies are carried out to evaluate the performance of the proposed method. As an illustration, the proposed method was applied to a real-life data that contains self-reported labor pain for women in two groups.
\end{abstract}


\end{frontmatter}

\section{Introduction}
\label{sec1}
Longitudinal data are very common in many areas of applied studies. Such data are repeatedly collected from independent subjects over time and correlation arises between measures from the same subject. One advantage of longitudinal study is that, additional to modeling the cohort effect, one can still specify the individual patterns of change. In order to take the correlation into consideration to not only avoid loss of efficiency in estimation but also make correct statistical inference, a number of methods are developed to evaluate covariate effects on the mean of a response variable \citep{LiangZeger1986, QuLindsayLi2000, JungYing2003}. \citet{Sutradhar2003} has proposed a generalization of the quasi-likelihood estimation approach to model the conditional mean of the response by solving the generalized quasi-likelihood (GQL) estimating equations. A general stationary auto-correlation matrix is used in this method, which, in fact, represents the correlations of many stationary dynamic models, such as stationary auto-regressive order 1 (AR(1)), stationary moving average order 1 (MA(1)), and stationary equi-correlation (EQC) models.

Quantile regression \citep{KoenkerBassett1978} has become a widely used technique in applications. The effects of covariates are modeled through conditional quantiles of the response variable, rather than the conditional mean, which makes it possible to characterize any arbitrary point of a distribution and thus provide a complete description of the entire response distribution. Compared to the classical mean regression, quantile regression is more robust to outliers and the error patterns do not need to be specified. Therefore, quantile regression has been widely used, \citep[see][among others]{ChenWeiParzen2004,Koenker2005,ReichBondellWang2010,Farcomeni2012}.

Recently quantile regression has been extended to longitudinal data analysis. A simple way to do so is to assume working independence that ignores correlations between repeated measures, which, of course, may cause loss of efficiency, see \citet{WeiHe2006,WangHe2007,MuWei2009,WangFygenson2009,Wang2009,WangZhuZhou2009}. \citet{Jung1996} firstly developed a quasi-likelihood method for median regression which incorporates correlations between repeated measures for longitudinal data. This method requires estimation of the correlation matrix. Based on Jung's work, \citet{LipsitzFMZ1997} proposed a weighted GEE model. \citet{Koenker2004} considered a random effect model for estimating quantile functions with subject specific fixed effects and based inference on a penalized likelihood. \citet{Karlsson2008} suggested a weighted approach for a nonlinear quantile regression estimation of longitudinal data. \citet{GeraciBottai2007} made inferences by using a random intercept to account for the within-subject correlations and proposed a method using the asymmetric Laplace distribution (ALD). \citeauthor{GeraciBottai2007}'s work was generalized by \citet{LiuBottai2009}, who gave a linear mixed effect quantile regression model using a multivariate Laplace distribution. \citet{Farcomeni2012} proposed a linear quantile regression model allowing time-varying random effects and modeled subject-specific parameters through a latent Markov chain. To reduce the loss of efficiency in inferences of quantile regression, \citet{TangLeng2011} incorporate the within-subject correlations through a specified conditional mean model.

Unlike in the classical linear regression, it is difficult to account for the correlations between repeated measures in quantile regression. Misspecification of the correlation structure in GEE method also leads to loss of inferential efficiency. Moreover, the approximating algorithms for computing estimates could be very complicated, and computational problems could occur when statistical software is applied to do intensive re-samplings in the inference procedure. To overcome these problems, \citet{FuWang2012} proposed a combination of the between- and within-subject estimating equations for parameter estimation. By combining multiple sets of estimating equations, \citet{LengZhang2012} developed a new quantile regression model which produces efficient estimates. Those two papers extend the induced smoothing method \citep{BrownWang2005} to quantile regression, and thus obtained smoothed objective functions which allow the application of Newton-Raphson iteration, and the latter automatically gives both the estimates of parameters and the sandwich estimate of their covariance matrix.

In this paper, we propose a more general quantile regression model by appropriately incorporating a correlation structure between repeated measures in longitudinal data. By employing a general stationary auto-correlation matrix, we avoid the specification of any particular correlation structure. The correlation coefficients can be iteratively estimated in the process of the regression estimation. By using the induced smoothed estimating functions, we can obtain estimates of parameters and their asymptotic covariance matrix by using Newton-Raphson algorithm. The estimators obtained using our proposed method are consistent and asymptotically normal. The results of the intensive simulation studies reveal that our proposed method outperforms those methods based on working independence assumption. Furthermore, our approach is simpler and more general than other quantile regression methods for longitudinal data on theoretical derivation, practical application and statistical programming.

The remainder of this paper proceeds as follows: In Section \ref{sec2} we develop the proposed quantile regression method and the algorithm of parameter estimation. The asymptotic properties of the proposed estimators are discussed in Section \ref{sec3}. Intensive simulation studies were carried out and results are presented in Section \ref{sec4}. Application of our method to the labor pain data is presented in Section \ref{sec5}. The paper is concluded in Section \ref{sec6} with some concluding remarks.

\section{Proposed quantile regression models}
\label{sec2}
Suppose, in a longitudinal setup, we collect a small number of repeated responses along with certain multidimensional covariates from a large number of independent individuals. Let $y_{i1},\dots,y_{ij},\dots,y_{in_{i}}$ denote $n_{i}\geq 2$ repeated measures observed from the $i$th subject, for $i=1,\dots,m$ where $m$ is a positive integer. Let $x_{ij}=(x_{ij1},\dots,x_{ijp})^{T}$ be the $p$-dimensional covariate vector corresponding to $y_{ij}$. Suppose that responses from different individuals are independent and those from the same subject are dependent. Let the conditional $\tau$th quantile of $y_{ij}$ given $x_{ij}$ be
\begin{equation*}
Q_{\tau}(y_{ij}|x_{ij})=x_{ij}^{T}\beta_{\tau}.
\end{equation*}
In quantile regression we are interested in estimating $\beta_{\tau}$ consistently and as efficiently as possible.

If we assume the working independence (WI) between repeated measures of responses among each individual, we can obtain $\hat{\beta}_{WI\tau}$, an estimate of $\beta_{\tau}$, by minimizing the following objective function
\begin{equation}\label{WIabjfn}
S(\beta_{\tau})=\sum_{i=1}^{m}\sum_{j=1}^{n_{i}}\rho_{\tau}(y_{ij}-x_{ij}^{T}\beta_{\tau}),
\end{equation}
where $\rho_{\tau}(u)=u(\tau-I(u\leq 0))$ \citep{KoenkerBassett1978}. Estimating equations can be derived from function (\ref{WIabjfn}) by equating the  differentiation of $S(\beta_{\tau})$ with respect to $\beta_{\tau}$ to 0. That is
\begin{equation}\label{WIeq}
U_{0}(\beta_{\tau})=\frac{\partial S(\beta_{\tau})}{\partial \beta_{\tau}}=\sum_{i=1}^{m}X_{i}^{T}\psi_{\tau}(y_{i}-X_{i}\beta_{\tau})=0,
\end{equation}
where $\psi_{\tau}(u)=\rho_{\tau}'(u)=\tau-I(u<0)$ is a discontinuous function, and $\psi_{\tau}(y_{i}-X_{i}\beta_{\tau})=(\psi_{\tau}(y_{i1}-x_{i1}^{T}\beta_{\tau}),\dots ,\psi_{\tau}(y_{in_{i}}-x_{in_{i}}^{T}\beta_{\tau}))^{T}$ is a $n_{i}\times 1$ vector. An efficient algorithm to obtain an estimate of  $\beta_{\tau}$ by solving the equation (\ref{WIeq}), $U_{0}(\beta_{\tau})=0$, was given by \citet{KoenkerD'Orey1987}, which is available in statistical software R (package quantreg). Parameter estimator $\hat{\beta}_{WI\tau}$ is derived from estimating equation (\ref{WIeq}) under working independence assumption, therefore the efficiency of $\hat{\beta}_{WI\tau}$ may not be satisfactory.

To take the within correlations into consideration when constructing quantile regression models for longitudinal data, a quasi-likelihood (QL) method was introduced by \citet{Jung1996}. Let $\varepsilon_{i}=(\varepsilon_{i1},\dots,\varepsilon_{ij},\dots,\varepsilon_{in_{i}})^{T}$, where $\varepsilon_{ij}=y_{ij}-x_{ij}^{T}\beta_{\tau}$ which is a continuous error term satisfying $P(\varepsilon_{ij}\leq 0)=\tau$ and with an unknown density function $f_{ij}(\cdot)$. In least squares, or mean regression model, Bernoulli distributed $\psi_{\tau}(\varepsilon_{i})=\psi_{\tau}(y_{i}-X_{i}\beta_{\tau})$ can be treated as a random noise vector. Using this fact, the QL can be generalized into quantile regression by estimating the correlation matrix of $\psi_{\tau}(\varepsilon_{i})$. Let the covariance matrix of $\psi_{\tau}(\varepsilon_{i})$ be denoted as
\begin{equation*}
V_{i}=\cov(\psi_{\tau}(y_{i}-X_{i}\beta_{\tau}))=\cov\begin{pmatrix}
\tau-I(\varepsilon_{i1}<0)\\
\vdots\\
\tau-I(\varepsilon_{in_{i}}<0)
\end{pmatrix},
\end{equation*}
and
\begin{equation*}
\varGamma_{i}=\diag[f_{i1}(0),\dots,f_{in_{i}}(0)]=\begin{pmatrix}
f_{i1}(0)& & \\
&\ddots&\\
& & f_{in_{i}}(0)
\end{pmatrix},
\end{equation*}
be an $n_{i}\times n_{i}$ diagonal matrix with $j$th diagonal element $f_{ij}(0)$. \citet{Jung1996} derived the derivative of the log-quasi-likelihood $l(\beta_{\tau};y_{i})$ with respect to $\beta_{\tau}$, which can be used to estimate $\beta_{\tau}$ by solving
\begin{equation}\label{u1eqn}
U_{1}(\beta_{\tau})=\sum_{i=1}^{m}\frac{\partial l(\beta_{\tau};y_{i})}{\partial\beta_{\tau}} =\sum_{i=1}^{m}X_{i}^{T}\varGamma_{i}V_{i}^{-1}\psi_{\tau}(y_{i}-X_{i}\beta_{\tau})=0.
\end{equation}
In estimating equation (\ref{u1eqn}), the term $\varGamma_{i}$ describes the dispersions in $\varepsilon_{ij}$ and its diagonal elements can be well estimated by following \citet{HendricksKoenker1992}:
\begin{equation*}
\hat{f}_{ij}(0)=2h_{n}[x_{ij}^{T}(\hat{\beta}_{\tau+h_{n}}-\hat{\beta}_{\tau-h_{n}})]^{-1},
\end{equation*}
where $h_{n}\rightarrow 0$ when $n\rightarrow \infty$ is a bandwidth parameter. In some cases when $f_{ij}$ is difficult to estimate, $\varGamma_{i}$ can be simply treated as an identity matrix with a slight loss of efficiency \citep{Jung1996}.

However, the estimation of the covariance matrix $V_{i}$ becomes much complicated when QL method is applied. Whatever correlation matrix that $\varepsilon_{i}$ follows, the correlation matrix of $\psi_{\tau}(\varepsilon_{i})$ is no longer the same one, and its correlation structure may be very difficult to specify.

Here, we propose a new method based on the following estimating equations
\begin{equation}\label{gqleeq}
U(\beta_{\tau})=\sum_{i=1}^{m}X_{i}^{T}\varGamma_{i}\varSigma_{i}^{-1}(\rho)\psi_{\tau}(y_{i}-X_{i}\beta_{\tau})=0,
\end{equation}
where $\varSigma_{i}(\rho)$ is the covariance matrix of $\psi_{\tau}(\varepsilon_{i})$ that can be expressed as $\varSigma_{i}(\rho)=A_{i}^{\frac{1}{2}}C_{i}(\rho)A_{i}^{\frac{1}{2}}$, with $A_{i}=\diag[\sigma_{i11},\dots,\sigma_{1n_{i}n_{i}}]$ being an $n_{i}\times n_{i}$ diagonal matrix, $\sigma_{ijj}=\var(\psi_{\tau}(\varepsilon_{ij}))$ and $C_{i}(\rho)$ as the correlation matrix of $\psi_{\tau}(\varepsilon_{i})$, $\rho$ being a correlation index parameter. Suppose the matrix $\varSigma_{i}(\rho)$ in equation (\ref{gqleeq}) has a general stationary auto-correlation structure such that the correlation matrix $C_{i}(\rho)$ takes the form of
\begin{equation*}
C_{i}(\rho)=
\begin{pmatrix}
1 & \rho_{1} & \rho_{2} & \cdots & \rho_{n_{i}-1} \\
\rho_{1} & 1 & \rho_{1} & \cdots & \rho_{n_{i}-2} \\
\vdots  & \vdots  & \vdots & & \vdots  \\
\rho_{n_{i}-1} & \rho_{n_{i}-2} & \rho_{n_{i}-3} & \cdots & 1
\end{pmatrix}
\end{equation*}
for all $i=1,\dots,m$, where $\rho_{\ell}$ can be estimated by
\begin{equation*}
\hat{\rho}_{\ell}=\frac{\sum_{i=1}^{m}\sum_{j=1}^{n_{i}-\ell}\tilde{y}_{ij}\tilde{y}_{i,j+\ell}/m(n_{i}-\ell)}{\sum_{i=1}^{m}\sum_{j=1}^{n_{i}}\tilde{y}_{ij}^{2}/mn_{i}}
\end{equation*}
for $\ell=1,\dots,n_{i}-1$ \citep{SutradharKova2000,Sutradhar2003} with $\tilde{y}_{ij}$ defined as
\begin{equation*}
\tilde{y}_{ij}=\frac{\psi_{\tau}(y_{ij}-x_{ij}^{T}\beta_{\tau})}{\sqrt{\sigma_{ijj}}}.
\end{equation*}
To estimate $\sigma_{ijj}=\var(\psi_{\tau}(y_{ij}-x_{ij}^{T}\beta_{\tau}))$, we apply the fact that $\psi_{\tau}(\varepsilon_{ij})=\psi_{\tau}(y_{ij}-x_{ij}^{T}\beta_{\tau})=\tau-I(y_{ij}<x_{ij}^{T}\beta_{\tau})$. Hence we have
\begin{equation*}
\begin{split}
\sigma_{ijj}&=\var[\psi_{\tau}(\varepsilon_{ij})]=\var[\tau-I(y_{ij}<x_{ij}^{T}\beta_{\tau})]=\var[I(y_{ij}<x_{ij}^{T}\beta_{\tau})]\\
& =\Pr(y_{ij}<x_{ij}^{T}\beta_{\tau})(1-\Pr(y_{ij}<x_{ij}^{T}\beta_{\tau})),
\end{split}
\end{equation*}
where $\Pr(y_{ij}<x_{ij}^{T}\beta_{\tau})$ is the probability of the event $\{y_{ij}<x_{ij}^{T}\beta_{\tau}\}$. If $\beta_{\tau}$ is the true parameter, we know that $x_{ij}^{T}\beta_{\tau}$ is exactly the $\tau$th quantile of the variable $y_{ij}$, hence $\Pr(y_{ij}<x_{ij}^{T}\beta_{\tau})=\tau$, which leads to an estimator of $\sigma_{ijj}$, \(\tilde{\sigma}_{ijj}=\tau(1-\tau)\).
Consequently, $A_{i}$ matrix can be calculated at the true $\beta_{\tau}$ as
\begin{equation}\label{Ai}
\begin{split}
\tilde{A}_{i}&=\diag[\tilde{\sigma}_{i11},\dots,\tilde{\sigma}_{1n_{i}n_{i}}]\\
&=\begin{pmatrix}
\tau(1-\tau)& & \\
 &\ddots& \\
 & & \tau(1-\tau)
\end{pmatrix}_{n_{i}\times n_{i}},
\end{split}
\end{equation}
indicating a constant diagonal matrix for a certain $\tau$. We denote the parameter estimator obtained from this proposed quantile regression (PQR) model as $\hat{\beta}_{PQR\tau}$.

Notice that in the expression $\varSigma_{i}(\rho)=A_{i}^{\frac{1}{2}}C_{i}(\rho)A_{i}^{\frac{1}{2}}$, if we set $A_{i}$ as the one at the true $\beta_{\tau}$ which is given by (\ref{Ai}),  $C_{i}(\rho)$ becomes the only part in $\varSigma_{i}(\rho)$ containing the information about the data and the parameter $\beta_{\tau}$. However in practice, the estimated parameter may never be exactly the true $\beta_{\tau}$. Thus, the elements of the diagonal matrix $A_{i}$ may differ from the constant value $\tau(1-\tau)$. Moreover, we expect the matrix $A_{i}$ to be also related to the parameter estimates, which becomes crucial when we use an iteration method to estimate parameters where the estimates $\hat{\beta}_{\tau}$ need to be updated within each iteration step. In this case, as long as the sample size is large enough, we can estimate the diagonal elements of $A_{i}$ by the following
\begin{equation*}
\begin{split}
\hat{\sigma}_{ijj}&=\Pr(y_{ij}<x_{ij}^{T}\beta_{\tau})(1-\Pr(y_{ij}<x_{ij}^{T}\beta_{\tau}))\\
&=\frac{1}{m}\sum_{i=1}^{m}I(y_{ij}<x_{ij}^{T}\beta_{\tau})(1-\frac{1}{m}\sum_{i=1}^{m}I(y_{ij}<x_{ij}^{T}\beta_{\tau})),
\end{split}
\end{equation*}
for all $j=1,\dots,n_{i}$ and $i=1,\dots,m$. By using $\hat{\sigma}_{ijj}$ to estimate $\varSigma_{i}$, the solution-finding iteration converges faster. The solution of estimating equations (\ref{gqleeq}) leads to an adjusted estimate of $\beta_{\tau}$. We call this method as adjusted quantile regression (AQR).

The difficulty of solving the estimating equation (\ref{gqleeq}) is caused by the non-convex and non-continuous objective function $U(\beta_{\tau})$ which is not differentiable. Though several methods can be applied to estimate $\beta_{\tau}$ from equation (\ref{gqleeq}) without requiring any derivatives and continuity of the estimating function, they may become very complicated and cause a high burden of computation. To overcome these difficulties, the induced smoothing method has been extended to the quantile regression for longitudinal data assuming a working correlation by \citet{FuWang2012}. The smoothing method is asymptotically equivalent to its original counterpart, see Lemma \ref{lemma} below.  Here, let $\tilde{U}(\beta_{\tau})=E_{Z}[U(\beta_{\tau}+\varOmega^{1/2}Z)]$, with expectation taken with respect to $Z$, where $Z\sim N(0,I_{p})$, and $\varOmega$ is updated as an estimate of the covariance matrix of parameter estimators. After some algebraic calculations, a smoothed estimating function $\tilde{U}(\beta_{\tau})$ is obtained as
\begin{equation}\label{smootheeq}
\tilde{U}(\beta_{\tau})=\sum_{i=1}^{m}X_{i}^{T}\varGamma_{i}\Sigma_{i}^{-1}(\rho)\tilde{\psi}_{\tau}(y_{i}-X_{i}\beta_{\tau})
\end{equation}
where \[
\tilde{\psi}_{\tau}= \begin{pmatrix}
\tau-1+\varPhi(\frac{y_{i1}-x_{i1}^{T}\beta_{\tau}}{r_{i1}})\\
\vdots\\
\tau-1+\varPhi(\frac{y_{in_{i}}-x_{in_{i}}^{T}\beta_{\tau}}{r_{in_{i}}})
\end{pmatrix}
\]
and $r_{ij}=\sqrt{x_{ij}^{T}\varOmega x_{ij}}$ for $j=1,\dots,n_{i}$. Thus the differentiation of $\tilde{U}(\beta_{\tau})$ with respect to $\beta_{\tau}$ can be easily calculated, and we can use $\partial\tilde{U}(\beta_{\tau})/\partial\beta_{\tau}$ as an approximation of $\partial U(\beta_{\tau})/\partial\beta_{\tau}$ as
\begin{equation*}
\frac{\partial\tilde{U}(\beta_{\tau})}{\partial\beta_{\tau}}=-\sum_{i=1}^{m}X_{i}^{T}\varGamma_{i}\Sigma_{i}^{-1}(\rho)\tilde{\varLambda}_{i}X_{i},
\end{equation*}
where $\tilde{\varLambda}_{i}$ is an $n_{i}\times n_{i}$ diagonal matrix with the $j$th diagonal element $\phi((y_{ij}-x_{ij}^{T}\beta_{\tau})/r_{ij})/r_{ij}$. Generally, let $\hat{\beta}_{WI\tau}$ be the estimate under the working independence assumption and $I_{p}$ be a identity matrix of size $p$, smoothed estimators of $\beta_{\tau}$ and its covariance matrix $\varOmega$ can be obtained from the following Newton-Raphson iteration:
\begin{enumerate}[Step 1.]
\item Given initial values of $\beta_{\tau}$ and the symmetric positive definite matrix $\varOmega$ as $\tilde{\beta}_{\tau}(0)=\hat{\beta}_{WI\tau}$ and $\tilde{\varOmega}(0)=\frac{1}{m}I_{p}$ respectively.
\item Using $\tilde{\beta}_{\tau}(r)$ and $\tilde{\varOmega}(r)$ given from the $r$th iteration, update $\tilde{\beta}_{\tau}(r+1)$ and $\tilde{\varOmega}(r+1)$ by
\begin{equation*}
\tilde{\beta}_{\tau}(r+1)=\tilde{\beta}_{\tau}(r)+\left[-\frac{\partial\tilde{U}(\beta_{\tau})}{\partial\beta_{\tau}}\right]_{r}^{-1}\times\biggl[\tilde{U}(\beta_{\tau})\biggr]_{r} \quad \text{and}
\end{equation*}
\begin{equation*}\label{Cov}
\tilde{\varOmega}(r+1)=\left[-\frac{\partial\tilde{U}(\beta_{\tau})}{\partial\beta_{\tau}}\right]_{r}^{-1}\times\biggl[\cov(\tilde{U}(\beta_{\tau}))\biggr]_{r}\times\left[-\frac{\partial\tilde{U}(\beta_{\tau})}{\partial\beta_{\tau}}\right]_{r}^{-1},
\end{equation*}
where $[]_{r}$ denotes that the expression between the square brackets is evaluated at $\beta_{\tau}=\tilde{\beta}_{\tau}(r)$ and $\cov(\tilde{U}(\beta_{\tau}))=\sum_{i=1}^{m}X_{i}^{T}\varGamma_{i}\Sigma_{i}^{-1}(\rho)\tilde{\psi}_{\tau}(\varepsilon_{i})\tilde{\psi}_{\tau}^{T}(\varepsilon_{i})\Sigma_{i}^{-1}(\rho)\varGamma_{i}X_{i}$.
\item Repeat step 2 until convergence.
\end{enumerate}
This method provides consistent estimates of $\beta_{\tau}$ and its covariance matrix $\varOmega$. Furthermore, compared with other techniques, our method based on Newton-Raphson algorithm is much faster.

\section{Asymptotic properties}
\label{sec3}

In this section, we derive asymptotic distributions of the proposed estimates obtained from both the original estimating equation (\ref{gqleeq}) and smoothed estimating equation (\ref{smootheeq}).
\begin{theorem}\label{thm1}
Under regularity conditions \ref{itm:A1}-\ref{itm:A5} listed in \ref{Appdix}, the estimator $\hat{\beta}_{\tau}$ based on the original estimating equation (\ref{gqleeq}) is $\sqrt{m}$-consistent and asymptotically normal,
\[
\sqrt{m}(\hat{\beta}_{\tau}-\beta_{\tau})\rightarrow N(0,G^{-1}(\beta_{\tau})V\{G^{-1}(\beta_{\tau})\}^{T}),
\]
where, in the variance-covariance matrix, $G(\beta_{\tau})=\lim_{m\rightarrow \infty}\frac{1}{m}\sum_{i=1}^{m}X_{i}^{T}\varGamma_{i}\Sigma_{i}^{-1}(\rho)\varGamma_{i}X_{i}$ and $V=\lim_{m\rightarrow \infty}\frac{1}{m}\sum_{i=1}^{m}X_{i}^{T}\varGamma_{i}\Sigma_{i}^{-1}(\rho)\cov\{\psi_{\tau}(y_{i}-X_{i}\beta_{\tau})\}\Sigma_{i}^{-1}(\rho)\varGamma_{i}X_{i}$.
\end{theorem}
\begin{lemma}\label{lemma}
Under regularity conditions \ref{itm:A1}-\ref{itm:A5} listed in \ref{Appdix}, the smoothed estimating functions $\tilde{U}(\beta_{\tau})$ are asymptotically equivalent to the original estimating functions $U(\beta_{\tau})$ in the sense that,
\[
\frac{1}{\sqrt{m}}\{\tilde{U}(\beta_{\tau})-U(\beta_{\tau})\}=o_{p}(1).
\]
\end{lemma}
\begin{theorem}\label{thm2}
Under regularity conditions \ref{itm:A1}-\ref{itm:A5} listed in \ref{Appdix}, the estimator $\tilde{\beta}_{\tau}$ based on the smoothed estimating equation (\ref{gqleeq}) is $\sqrt{m}$-consistent and asymptotically normal,
\[
\sqrt{m}(\tilde{\beta}_{\tau}-\beta_{\tau})\rightarrow N(0,G^{-1}(\beta_{\tau})V\{G^{-1}(\beta_{\tau})\}^{T}),
\]
where $G(\beta_{\tau})$ and $V$ have the same expressions as in Theorem \ref{thm1}.
\end{theorem}

 Note that, Lemma \ref{lemma} indicates the asymptotic equivalence of the smoothed estimating functions and their original counterpart. Theorems \ref{thm1} and \ref{thm2} illustrate the asymptotic equivalence of the two corresponding estimators. From Theorem \ref{thm2}, we can obtain a natural sandwich form estimator of the variance-covariance matrix of $\sqrt{m}(\tilde{\beta}_{\tau}-\beta_{\tau})$ as
\begin{equation}\label{???}
\widehat{\cov}(\sqrt{m}(\tilde{\beta}_{\tau}-\beta_{\tau}))=\hat{G}^{-1}(\tilde{\beta}_{\tau})\tilde{V}\{\hat{G}^{-1}(\tilde{\beta}_{\tau})\}^{T},
\end{equation}
where in the covariance matrix, we have $\tilde{G}(\tilde{\beta}_{\tau})=\frac{1}{m}\sum_{i=1}^{m}X_{i}^{T}\varGamma_{i}\Sigma_{i}^{-1}(\rho)\varGamma_{i}X_{i}$ and $\tilde{V}=\frac{1}{m}\sum_{i=1}^{m}X_{i}^{T}\varGamma_{i}\Sigma_{i}^{-1}(\rho)\cov\{\tilde{\psi}_{\tau}(y_{i}-X_{i}\tilde{\beta}_{\tau})\}\Sigma_{i}^{-1}(\rho)\varGamma_{i}X_{i}$. Based on this formula, we update matrix $\tilde{\varOmega}$ in the Newton-Raphson iteration on page \pageref{Cov}.

Proofs are deferred to the \ref{Appdix}.

\section{Simulation studies}
\label{sec4}

\begin{table}[t!]
\caption[Simulation Results with Normal Errors]{Biases and relative efficiencies to the estimators of $\beta_{0}$, $\beta_{1}$ and $\beta_{2}$ using different methods at quantiles 0.25, 0.5, 0.95 are reported for Case 1 when $\rho=0.1$, 0.5, 0.9.}
\label{NorEr}
\begin{tabular*}{\textwidth}{@{\extracolsep{\fill}}llrrcrcrcrcrcr@{}}
\toprule
& & & \multicolumn{3}{c}{$\beta_{0}$} & &\multicolumn{3}{c}{$\beta_{1}$} & & \multicolumn{3}{c}{$\beta_{2}$}\\
\cmidrule{4-6}\cmidrule{8-10}\cmidrule{12-14}
$\tau$ & $\rho$ & Method & Bias & & EFF & & Bias & & EFF & & Bias & & EFF \\ \midrule
0.25 & 0.1 & AQR &  0.0019 && 1.046 &&  0.0015 && 1.041 && -0.0001 && 1.069   \\
         &  & PQR & -0.0005 && 1.040 &&  0.0014 && 1.044 && -0.0001 && 1.069   \\
          &  & WI &  0.0000 && 1.000 &&  0.0020 && 1.000 && -0.0002 && 1.000   \\
     & 0.5 & AQR &  0.0043 && 1.098 && -0.0017 && 1.191 &&  0.0015 && 1.236   \\
         &  & PQR &  0.0017 && 1.111 && -0.0017 && 1.194 &&  0.0015 && 1.242   \\
          &  & WI &  0.0025 && 1.000 && -0.0019 && 1.000 &&  0.0016 && 1.000   \\
     & 0.9 & AQR &  0.0048 && 1.186 && -0.0006 && 2.811 &&  0.0000 && 2.707   \\
         &  & PQR &  0.0013 && 1.195 && -0.0008 && 2.816 &&  0.0000 && 2.706   \\
          &  & WI &  0.0028 && 1.000 && -0.0023 && 1.000 &&  0.0006 && 1.000   \\
 0.5 & 0.1 & AQR &  0.0022 && 1.050 && -0.0020 && 1.054 &&  0.0006 && 1.049   \\
         &  & PQR &  0.0022 && 1.050 && -0.0020 && 1.053 &&  0.0006 && 1.049   \\
          &  & WI &  0.0019 && 1.000 && -0.0017 && 1.000 &&  0.0008 && 1.000   \\
     & 0.5 & AQR &  0.0002 && 1.059 &&  0.0018 && 1.260 &&  0.0010 && 1.247   \\
         &  & PQR &  0.0002 && 1.059 &&  0.0018 && 1.260 &&  0.0010 && 1.247   \\
          &  & WI & -0.0000 && 1.000 &&  0.0025 && 1.000 &&  0.0006 && 1.000   \\
     & 0.9 & AQR & -0.0003 && 1.256 && -0.0005 && 3.135 &&  0.0001 && 3.026   \\
         &  & PQR & -0.0003 && 1.256 && -0.0005 && 3.136 &&  0.0001 && 3.026   \\
          &  & WI & -0.0004 && 1.000 && -0.0013 && 1.000 &&  0.0015 && 1.000   \\
0.95 & 0.1 & AQR & -0.0094 && 1.066 &&  0.0047 && 1.069 &&  0.0029 && 1.099   \\
         &  & PQR &  0.0033 && 1.092 &&  0.0046 && 1.071 &&  0.0028 && 1.118   \\
          &  & WI & -0.0004 && 1.000 &&  0.0032 && 1.000 &&  0.0031 && 1.000  \\
     & 0.5 & AQR & -0.0142 && 0.988 && -0.0008 && 1.136 && -0.0018 && 1.257   \\
         &  & PQR &  0.0001 && 1.092 && -0.0016 && 1.144 && -0.0015 && 1.248   \\
          &  & WI & -0.0039 && 1.000 && -0.0039 && 1.000 && -0.0017 && 1.000   \\
     & 0.9 & AQR & -0.0136 && 1.212 && -0.0003 && 2.152 && -0.0015 && 2.187   \\
         &  & PQR &  0.0037 && 1.244 && -0.0001 && 2.155 && -0.0017 && 2.129   \\
          &  & WI & -0.0021 && 1.000 &&  0.0052 && 1.000 && -0.0002 && 1.000   \\
\bottomrule
\end{tabular*}
\end{table}

In order to examine the small sample performance of the proposed method, we conducted extensive simulation studies. A part of the simulation results are reported in this section.

The random samples are generated from the model
\begin{equation}
y_{ij}=\beta_{0}+x_{ij1}\beta_{1}+x_{ij2}\beta_{2}+\varepsilon_{ij}
\end{equation}
for $i=1,\dots,m$ and $j=1,\dots,n_{i}$, where $x_{ij1}$ are sampled from the Bernoulli distribution with probability $0.5$, $Bernoulli(0.5)$, and $x_{ij2}$ are generated from a standard normal distribution. In this simulation study, we set the sample size $m=500$ and a balanced design $n_{i}=4$ for all $i=1,\dots,500$. Let the variance-covariance matrix of $\varepsilon_{i}$ with an AR(1) structure be expressed as
\[
\varSigma_{\varepsilon}(\rho)=\begin{pmatrix}
1 & \rho & \rho^2 & \dots & \rho^{n_{i}-1}\\
\rho & 1 & \rho & \dots & \rho^{n_{i}-2}\\
 &  &  & \vdots & \\
\rho^{n_{i}-1} & \rho^{n_{i}-2} &  & \dots & 1\\
\end{pmatrix},
\]
where $\rho$ is set to be $0.1$, $0.5$, or $0.9$ respectively, to generate errors with low, medium and high correlation. Three different distributions are considered for the random error $\varepsilon_{i}$:
\begin{enumerate}[\textit{Case} 1.]
\item Normal distribution, assume that $\varepsilon_{i}$ follows a multivariate normal distribution with mean $-q_{\tau}$, or the $ \tau$th quantile of $0$ and covariance $\varSigma_{\varepsilon}(\rho)$, $N_{p}(-q_{\tau},\varSigma_{\varepsilon}(\rho))$, where $q_{\tau}$ is the $\tau$th quantile of the standard normal distribution.
\item Chi-squared distribution, assume that $ \varepsilon_{i}= \varepsilon_{i}'-q_{\tau}$, where $ \varepsilon_{i}'$ follows a multivariate Chi-squared distribution with two degrees of freedom $(\chi_{2}^2)$, where $q_{\tau}$ is the $\tau$th quantile of the $\chi_{2}^2$ distribution.
\item Student's T distribution, suppose that   $ \varepsilon_{i}= \varepsilon_{i}'-q_{\tau}$, where $ \varepsilon_{i}'$  follows a multivariate T distribution with three degrees of freedom $(T_{3})$, where $q_{\tau}$ is the $\tau$th quantile of the $T_{3}$ distribution.
\end{enumerate}
The values of parameters used in the simulation are $\beta_{0}=-0.5$, $\beta_{1}=0.5$ and $\beta_{2}=1$. Quantiles of $\tau=0.25$, 0.5 and 0.95 are chosen to study the performance of the quantile regression estimators for the response distribution.

The results of $ 1,000$ simulation runs of quantile regression using different estimation methods are analyzed. We report the average bias ($Bias$) and relative efficiency ($EFF$) of the estimates of $\beta_{0}$, $\beta_{1}$ and $\beta_{2}$ using different quantile regression methods (quantile regression method assuming working independence (WI), proposed quantile regression method (PQR), and adjusted quantile regression method (AQR)) in the attached Tables. For each estimator, we use SD to denote the standard deviation of $1000$ parameter estimates, SE the average of $1000$ estimated standard errors. For our proposed estimators, $P_{0.95}$ denotes the percentage of simulation runs when the true parameter falls into the $95\% $ confidence intervals constructed based on the sandwich estimate of the covariance matrix of $\hat{\beta}_{\tau}$, at quantiles $0.25, 0.5, 0.95$. Where $\rho$ is specified as $ 0.1, 0.5$, and $0.9$ respectively.

\begin{table}[t!]
\caption{Simulation Results with Normal Errors (\textit{case} 1).}
\label{SDNor}
\begin{tabular*}{\columnwidth}{@{\extracolsep{\fill}}llrrrrrrrrrr@{}}
\toprule
& & & \multicolumn{3}{c}{$\beta_{0}$} & \multicolumn{3}{c}{$\beta_{1}$} & \multicolumn{3}{c}{$\beta_{2}$}\\
\cmidrule{4-6}\cmidrule{7-9}\cmidrule{10-12}
$\tau$ & $\rho$ & Method & SD & SE & $P_{0.95}$ &  SD & SE & $P_{0.95}$ &  SD & SE & $P_{0.95}$ \\ \midrule
0.25 & 0.1 & AQR & 0.043 & 0.042 & 0.951 & 0.060 & 0.060 & 0.949 & 0.029 & 0.030 &  0.950  \\
         &  & PQR & 0.043 & 0.042 & 0.953 & 0.060 & 0.059 & 0.951 & 0.029 & 0.030 &  0.949  \\
     & 0.5 & AQR & 0.047 & 0.046 & 0.947 & 0.055 & 0.055 & 0.947 & 0.028 & 0.027 &  0.952  \\
         &  & PQR & 0.047 & 0.046 & 0.950 & 0.054 & 0.055 & 0.946 & 0.027 & 0.028 &  0.952  \\
     & 0.9 & AQR & 0.054 & 0.053 & 0.944 & 0.038 & 0.036 & 0.939 & 0.019 & 0.018 &  0.949  \\
         &  & PQR & 0.054 & 0.053 & 0.945 & 0.038 & 0.036 & 0.939 & 0.019 & 0.018 &  0.947  \\
0.5  & 0.1 & AQR & 0.040 & 0.040 & 0.943 & 0.055 & 0.055 & 0.945 & 0.027 & 0.027 &  0.956  \\
         &  & PQR & 0.040 & 0.040 & 0.943 & 0.055 & 0.055 & 0.945 & 0.027 & 0.027 &  0.956  \\
     & 0.5 & AQR & 0.044 & 0.043 & 0.948 & 0.050 & 0.049 & 0.949 & 0.025 & 0.024 &  0.955  \\
         &  & PQR & 0.044 & 0.043 & 0.948 & 0.050 & 0.049 & 0.949 & 0.025 & 0.024 &  0.955  \\
     & 0.9 & AQR & 0.045 & 0.050 & 0.947 & 0.032 & 0.033 & 0.948 & 0.017 & 0.016 &  0.950  \\
         &  & PQR & 0.045 & 0.050 & 0.947 & 0.032 & 0.033 & 0.948 & 0.017 & 0.016 &  0.950  \\
0.95 & 0.1 & AQR & 0.066 & 0.062 & 0.926 & 0.090 & 0.088 & 0.945 & 0.046 & 0.043 &  0.934  \\
         &  & PQR & 0.066 & 0.063 & 0.933 & 0.090 & 0.089 & 0.944 & 0.045 & 0.044 &  0.939  \\
     & 0.5 & AQR & 0.070 & 0.065 & 0.924 & 0.090 & 0.084 & 0.923 & 0.044 & 0.041 &  0.935  \\
         &  & PQR & 0.068 & 0.066 & 0.939 & 0.089 & 0.085 & 0.931 & 0.044 & 0.042 &  0.943  \\
     & 0.9 & AQR & 0.080 & 0.079 & 0.926 & 0.063 & 0.058 & 0.919 & 0.031 & 0.028 &  0.928  \\
         &  & PQR & 0.080 & 0.079 & 0.920 & 0.063 & 0.059 & 0.925 & 0.031 & 0.028 &  0.929  \\
\bottomrule
\end{tabular*}
\end{table}

Table \ref{NorEr} shows the results when $\varepsilon_{i}$ follows a multivariate normal distribution (\textit{Case} 1) with an AR(1) correlation structure where the value of $\rho$ is specified as 0.1, 0.5 and 0.9 respectively. As we can see, when the correlation is low ($\rho=0.1$), the average biases and relative efficiencies of quantile regression estimators $\hat{\beta}_{PQR\tau}$ and $\hat{\beta}_{AQR\tau}$ are comparable, and these two estimators perform slightly better than the quantile regression estimator assuming working independence ($\hat{\beta}_{WI\tau}$). When the correlation is higher ($\rho=0.5$, or $\rho=0.9$), the proposed estimators $\hat{\beta}_{PQR\tau}$ and $\hat{\beta}_{AQR\tau}$ are equally efficient with small biases and much smaller variances than the working independence estimator. Moreover, the estimators $\hat{\beta}_{PQR\tau}$ and $\hat{\beta}_{AQR\tau}$ become more efficient as the correlation ($\rho$) increases. In general, these two proposed methods provide more efficient estimates of $\beta_{1\tau}$ and $\beta_{2\tau}$ than the intercept parameter $\beta_{0\tau}$. Similar performances are observed when $\varepsilon_{i}$ is  $\chi_{2}^2$ (\textit{case} 2) or  $T_{3}$ (\textit{case} 3) distributed except that the proposed estimators are more efficient at higher quantiles when the random effect follows a $\chi_{2}^2$ distribution (\textit{case} 2). The results are not reported.

 Another observation was made to compare the sample standard deviation ($SD$) and the average asymptotic standard errors ($SE$) of the proposed and adjusted estimators when $\varepsilon_{i}$ is normally distributed. In Table \ref{SDNor} we can see that each value of $SD$ is very small and close to the corresponding $SE$ value, which means our estimators perform very well and the estimate of the standard deviation of $\hat{\beta}_{\tau}$ works very well also. Furthermore, the percentages of simulation runs ($P_{0.95}$) when the true parameters fall into the 95\% confidence intervals are all very close to their nominal level, evidencing the asymptotic normality of the estimators. Hence inferences based on it are reliable. The results for Cases 2 and 3 are similar.

Simulation results comparing the linear mixed effects model (LME) and the proposed median regression models are reported in Table \ref{LmeMd}. Biases($Bias$) and relative efficiencies($EFF$) to each estimator are reported for three different error distributions(\textit{case} 1, 2 and 3). 
\begin{table}[t!]
\caption{Simulation Results of Linear Mixed Effect Model and Median Regression Models} \label{LmeMd}
\begin{tabular*}{\textwidth}{@{\extracolsep{\fill}}llrrcrcrcrcrcr@{}}
\toprule
& & & \multicolumn{3}{c}{$\beta_{0}$} & &\multicolumn{3}{c}{$\beta_{1}$} & & \multicolumn{3}{c}{$\beta_{2}$}\\
\cmidrule{4-6}\cmidrule{8-10}\cmidrule{12-14}
Err & $\rho$ & Method & Bias & & EFF & & Bias & & EFF & & Bias & & EFF \\ \midrule
 Nor & 0.1 &  LME &  0.0015 && 1.542 && -0.0020 && 1.546 &&  0.0004 && 1.543   \\
         &  & PQR &  0.0022 && 1.050 && -0.0020 && 1.053 &&  0.0006 && 1.049   \\
          &  & WI &  0.0019 && 1.000 && -0.0017 && 1.000 &&  0.0008 && 1.000   \\
     & 0.5 &  LME &  0.0008 && 1.517 &&  0.0022 && 2.088 &&  0.0006 && 2.048   \\
         &  & PQR &  0.0002 && 1.059 &&  0.0018 && 1.260 &&  0.0010 && 1.247   \\
          &  & WI & -0.0000 && 1.000 &&  0.0025 && 1.000 &&  0.0006 && 1.000   \\
     & 0.9 &  LME & -0.0021 && 1.729 &&  0.0007 && 7.979 &&  0.0001 && 7.605   \\
         &  & PQR & -0.0003 && 1.256 && -0.0005 && 3.136 &&  0.0001 && 3.026   \\
          &  & WI & -0.0004 && 1.000 && -0.0013 && 1.000 &&  0.0015 && 1.000   \\
 Chi & 0.1 &  LME &  0.6193 && 0.011 && -0.0024 && 1.008 &&  0.0005 && 1.002   \\
         &  & PQR &  0.0084 && 1.030 && -0.0059 && 1.037 &&  0.0015 && 1.042   \\
          &  & WI &  0.0069 && 1.000 && -0.0060 && 1.000 &&  0.0018 && 1.000   \\
     & 0.5 &  LME &  0.6159 && 0.011 &&  0.0014 && 1.148 && -0.0005 && 1.020   \\
         &  & PQR &  0.0064 && 1.033 &&  0.0008 && 1.061 && -0.0004 && 1.084   \\
          &  & WI &  0.0051 && 1.000 &&  0.0006 && 1.000 &&  0.0001 && 1.000   \\
     & 0.9 &  LME &  0.6149 && 0.018 &&  0.0003 && 2.802 && -0.0002 && 2.406   \\
         &  & PQR &  0.0053 && 1.188 && -0.0019 && 1.982 &&  0.0004 && 1.864   \\
          &  & WI &  0.0024 && 1.000 &&  0.0010 && 1.000 &&  0.0008 && 1.000   \\
   T & 0.1 &  LME &  0.0016 && 0.613 &&  0.0001 && 0.650 && -0.0007 && 0.631   \\
         &  & PQR &  0.0004 && 1.031 && -0.0007 && 1.052 && -0.0001 && 1.046   \\
          &  & WI &  0.0005 && 1.000 && -0.0009 && 1.000 &&  0.0001 && 1.000   \\
     & 0.5 &  LME &  0.0008 && 0.583 && -0.0002 && 0.849 && -0.0009 && 0.831   \\
         &  & PQR &  0.0006 && 1.122 &&  0.0003 && 1.234 &&  0.0004 && 1.301   \\
          &  & WI &  0.0005 && 1.000 &&  0.0009 && 1.000 &&  0.0001 && 1.000   \\
     & 0.9 &  LME &  0.0034 && 0.591 && -0.0012 && 2.863 &&  0.0006 && 3.338   \\
         &  & PQR &  0.0020 && 1.223 && -0.0003 && 2.916 &&  0.0002 && 2.740   \\
          &  & WI &  0.0031 && 1.000 && -0.0015 && 1.000 &&  0.0010 && 1.000   \\
\bottomrule
\end{tabular*}
\end{table}

As we have expected, quantile (Median) regression outperforms mean regression when the random error distribution is skewed or heavy-tailed. When the error term follows a normal distribution, the LME and the proposed quantile method have comparable bias, but the LME is more efficient than the median regression according to the average of the estimated efficiencies of the three $\beta_{\tau}$-parameters. However, when the error follows chi-square distribution ($\chi_{2}^2$) or student's t distribution ($T_{3}$), the LME performs worse than our proposed median regression method, particularly in estimating the intercept parameter $\beta_{0\tau}$. The median regression model is more robust to model mis-specification, while LME can only provide misleading results in those cases.

\section{A real data example}
\label{sec5}
\begin{figure}
\begin{center}
\includegraphics [height=90mm, width=150mm]{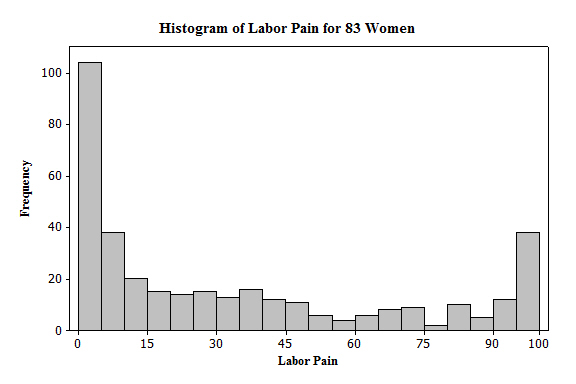}
\caption[Histogram of Measured Labor Pain]{Histogram of measured labor pain for all 83 women.}
\label{histpain}
\end{center}
\end{figure}

In this section, we illustrate the proposed method for quantile regression by analyzing the labor pain data, reported by \citet{Davis1991} and successfully analyzed by \citet{Jung1996}. The data set arose from a randomized clinical trial on the effectiveness of a medication for labor pain relief. A total of $m=83$ women were randomly assigned to either a pain medication group ($43$ women) or a placebo group ($40$ women). The response is a self-reported measure of pain measured every $30$ minutes on a $100$-mm line, where $0=$ ``no pain" and $100=$ ``extreme unbearable pain". The maximum number of measures for each women was $6$, but at later measurement times there are numerous values missing with a nearly monotone pattern. In Figure \ref{histpain}, a histogram of all the pains shows that the data is severely skewed. Therefore mean regression may not be appropriate. In Figure \ref{boxpain}, a box-plot shows the mean and median of the pain over time for all $83$ women and those in two different groups. Statistical dependence on the temporal course of the quartiles of the response is evident to some extent, especially for the placebo group.

\begin{figure}[t!]
\begin{center}
\includegraphics [height=90mm, width=150mm]{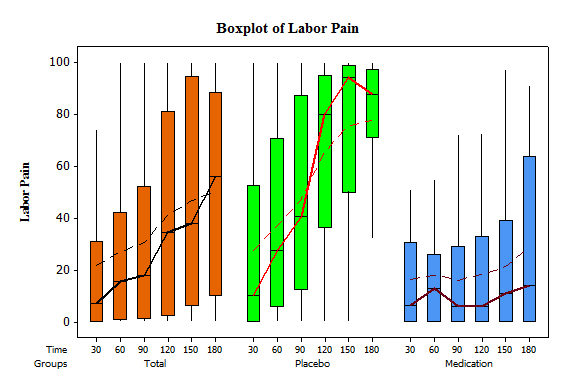}
\caption[Box-plot of Measured Labor Pain]{Box-plot of measured labor pain for all women in placebo and medication groups. The thick solid lines represent the median, while the means are connected with thin dashed lines.}
\label{boxpain}
\end{center}
\end{figure}

Let $y_{ij}$ be the amount of pain for the $i$th patient at time $j$,  $R_{i}$ be the treatment indicator taking $0$ for placebo and $1$ for medication, and $T_{ij}$ be the measurement time divided by $30$ minutes. \citet{Jung1996} considered the median regression model
\begin{equation}
y_{ij}=\beta_{0}+\beta_{1}R_{i}+\beta_{2}T_{ij}+\beta_{3}R_{i}T_{ij}+\varepsilon_{ij},
\end{equation}
where $\varepsilon_{ij}$ is a zero-median error term. Note that $(\beta_{0}+\beta_{1})+(\beta_{2}+\beta_{3})T_{ij}$ is the median for the treatment group and the median for the placebo group is $\beta_{0}+\beta_{2}T_{ij}$.

\begin{table}
\caption[Estimated Results of Labor Pain Data]{Estimated parameters ($\widehat{\beta}$), their standard errors (SE) and corresponding $95\%$ confidence intervals (CI) from fitting both the proposed quantile regression model and usual quantile regression assuming working independence at three quartiles, $\tau= 0.25, 0.5, 0.75$.}
\label{Pod_WI_LP}
\begin{tabular*}{\textwidth}{@{\extracolsep{\fill}}llccrcr@{, }rcrcr@{, }r@{}}
\toprule
&& && \multicolumn{4}{c}{Proposed Method} & &\multicolumn{4}{c}{WI} \\
\cmidrule{5-8}\cmidrule{10-13}
$\tau$ && $\beta$ && $\widehat{\beta}$ & SE & \multicolumn{2}{c}{CI} & & $\widehat{\beta}$ & SE & \multicolumn{2}{c}{CI} \\ \midrule
0.25 && $\beta_{0}$ && -10.32 & 0.42 & (-11.13 & -9.50)   & &  -10.83 &  2.20 & (-15.14 & -6.52) \\
     && $\beta_{1}$ &&   9.08 & 0.42 &   (8.27 & 9.90)    & &   10.83 &  2.20 & (6.51   & 15.15) \\
     && $\beta_{2}$ &&  17.72 & 0.41 & (16.92  & 18.51)   & &   10.83 &  2.20 & (6.52   & 15.14) \\
     && $\beta_{3}$ && -15.58 & 0.41 & (-16.38 & -14.79)  & &  -10.83 &  2.20 & (-15.15 & -6.51) \\
0.5  && $\beta_{0}$ && -10.44 & 1.54 & (-13.45 & -7.43)   & &   -6.20 &  7.95 & (-21.77 & 9.37) \\
     && $\beta_{1}$ &&   8.96 & 1.54 & (5.95   & 11.97)   & &   12.20 &  8.88 & (-5.21  & 29.61) \\
     && $\beta_{2}$ &&  21.05 & 1.27 & (18.56  & 23.53)   & &   17.20 &  2.35 & (12.60  & 21.80) \\
     && $\beta_{3}$ && -12.25 & 1.27 & (-14.74 & -9.77)   & &  -16.20 &  2.72 & (-21.53 & -10.87) \\
0.75 && $\beta_{0}$ &&   1.02 & 4.08 & (-6.97  & 9.02)    & &   58.67 & 14.83 & (29.60  & 87.74) \\
     && $\beta_{1}$ &&  20.42 & 4.08 & (12.43  & 28.42)   & &  -42.67 & 16.30 & (-74.61 & -10.72) \\
     && $\beta_{2}$ &&  22.84 & 0.68 & (21.51  & 24.17)   & &    7.67 &  3.44 & (0.93   & 14.40) \\
     && $\beta_{3}$ && -10.46 & 0.68 & (-11.79 & -9.13)   & &   -2.67 &  4.02 & (-10.54 & 5.21) \\
\bottomrule
\end{tabular*}
\end{table}

Our proposed quantile regression model was fit for three quartiles, $\tau=$ $0.25, 0.5$ and $0.75$, respectively. We report the estimated parameters ($\widehat{\beta}$), their asymptotic standard errors ($SE$) and the $95\%$ confidence intervals ($CI$) in Table \ref{Pod_WI_LP}. Here we also list the results of the usual quantile regression method assuming working independence for comparison. At the $0.25$th quantile, we see that our proposed method gives smaller standard errors, although these two methods produce comparable estimates of parameters. Note that all parameter estimates are significant at $5\%$ level, meaning that each covariate has effect on the $25\%$ quantile labor pain. Parameter estimates to the median regression methods have similar properties, except that the usual quantile regression method assuming working independence gives insignificant estimates of $\beta_{0}$ and $\beta_{1}$, indicating similar baseline pain among two groups. While, for the third quartile ($0.75$th quantile), our proposed method and the WI method  have very different parameter estimates with the proposed method giving much smaller standard errors of the estimates. The insignificant $\beta_{3}$ in WI method indicates similar time effects on the amount of pain in groups of placebo and medication, which contradicts our medical knowledge, while the significance of $\beta_{3}$ in our proposed method provides a perfect interpretation.

To investigate how treatment and time affect the amount of labor pain at three quartiles ($0.25, 0.5, 0.75$), we use our proposed method to compare the estimated values of $\beta_{0}$ with $\beta_{0}+\beta_{1}$ and $\beta_{2}$ with $\beta_{2}+\beta_{3}$ at each quartile, respectively. The result is visualized in Figure \ref{ETPain}. In Figure \ref{ETPain}, we can easily see that medication treatment do help women relieve their labor pain, and the pain of women in the placebo group grows faster with time than that in the treatment group. Moreover, the amount of pain tends to grow slightly faster at higher quantiles than that at lower quantiles. These conclusions are consistent with the box plots shown in Figure \ref{boxpain} and results in \citet{Jung1996} and \citet{LengZhang2012}.

\begin{figure}[t!]
\begin{center}
\includegraphics [height=90mm, width=150mm]{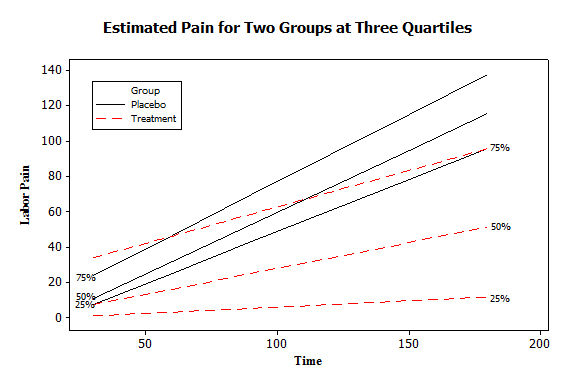}
\caption[Estimated Labor Pain Using Quantile Regression]{Labor pain obtained by using proposed quantile regression method at three quartiles 25\%, 50\% and 75\%.}
\label{ETPain}
\end{center}
\end{figure}

\section{Conclusion}
\label{sec6}
In this paper, we have proposed a new quantile regression model for longitudinal data, incorporating the correlations between repeated measures. We applied a general stationary auto-correlation structure to the estimating equations. To reduce the computational burden caused by the non-continuous estimating functions, we have employed the induced smoothing method of \citet{FuWang2012} for quantile regression. The estimates of the regression parameters and their covariance matrix are then obtained using Newton-Raphson iteration technique. It can be seen that our proposed method is a simple and efficient way to account for within-subject correlations in quantile regression for longitudinal data. This approach drew the inferential methods of quantile regression and the classical mean regression much closer. It reveals that the techniques in GEE's are applicable in quantile regression modeling. Our simulation studies indicate that the proposed method performs better than other methods assuming working independence especially when the within correlation is high. Furthermore, a comparison is also made between the proposed median regression estimator and the corresponding mean regression estimator, where the former is found to be better in analyzing heavy-tailed or skewed data. Finally, the proposed quantile regression estimator is applied to a real data set where the labor pain of two groups of women are reported, which reveals how treatment and time affect the amount of labor pain at three quartiles.

We were trying to take the within-subject correlation into consideration of quantile regression modeling, while the effects of unobserved covariates which may be different from individual to individual have not been captured. For instance, in our real data application, the personal perception of labor pain may vary from one to another. Therefore, like what has been done in \citet{Koenker2004}, we may extend our proposed model to a penalized version allowing individual specific effects by adding subject specific parameters and a penalty term. Further developments of our proposed method include extending quantile regression to well studied research areas in mean regression for longitudinal data such as mixed models for count and binary data \citep{Sutradhar2011}, nonlinear models \citep{HeFuFung2003}, semi-parametric models \citep{LinCarroll2006}, and nonparametric models \citep{WuZhang2006,QuLi2006}. Further results will be reported in forthcoming papers.

\renewcommand{\thesection}{Acknowledgement}
\section{}\label{Acknowledge}
\addcontentsline{toc}{section}{Acknowledgement}

The research of the authors was partially supported by a Discovery Grant from the Natural Sciences and Engineering Research Council of Canada.

\bibliographystyle{elsarticle-harv}
\bibliography{refs}

\renewcommand{\thesection}{Appendix}
\section{}\label{Appdix}
\addcontentsline{toc}{section}{Appendix}
\appendix

In the appendix we give a set of regularity conditions and outline the proofs of the theorems in Section \ref{sec3}.
\begin{enumerate}
\renewcommand{\theenumi}{A\arabic{enumi}}
\item\label{itm:A1} For each $i$, the number of repeated measures $n_{i}$ is bounded and the dimension $p$ of covariates $x_{ij}$ is fixed. The cumulative distribution functions $F_{ij}(z)=P(y_{ij}-x_{ij}^{T}\beta_{\tau}\leq z|x_{ij})$ are absolutely continuous, with continuous densities $f_{ij}$ and its first derivative being uniformly bounded away from 0 and $\infty$ at the point 0, $i=1,\dots,m$; $j=1,\dots,n_{i}$.
\item\label{itm:A2} The true value $\beta_{\tau}$ is an interior point of a bounded convex region $\mathfrak{B}$.
\item\label{itm:A3} Each $x_{i}$ satisfies the following conditions
\begin{enumerate}[(a)]
\item For any positive definite matrix $W_{i}$, $\frac{1}{m}\sum_{i=1}^{m}X_{i}^{T}W_{i}\varGamma_{i}X_{i}$ converges to a positive definite matrix; where $\varGamma_{i}$ is an $n_{i}\times n_{i}$ diagonal matrix with the $j$th diagonal element $f_{ij}(0)$.
\item $\sup_{i}\| x_{i}\|<+\infty$, where $\|\cdot\|$ denotes the Euclidean norm.
\end{enumerate}
\item\label{itm:A4} Matrix $\varOmega$ is positive definite and $\varOmega=O(\frac{1}{m})$.
\item\label{itm:A5} The differentiation of negative $\tilde{U}(\beta_{\tau})$, $-\partial \tilde{U}(\beta_{\tau})/\partial\beta_{\tau}$, is positive definite with probability 1.
\end{enumerate}

\begin{proof}[Proof of Theorem \ref{thm1}]
Let $H_{i}^{T}=X_{i}^{T}\varGamma_{i}\Sigma_{i}^{-1}(\rho)$ and $\psi_{i}=\psi_{\tau}(y_{i}-X_{i}\hat{\beta}_{\tau})$, therefore $U(\hat{\beta}_{\tau})=\sum_{i=1}^{m}H_{i}^{T}\psi_{i}$. Let $\bar{U}(\hat{\beta}_{\tau})=\sum_{i=1}^{m}H_{i}^{T}\varphi_{i}$, where $\varphi_{i}=(\tau-P(y_{i1}-x_{i1}^{T}\hat{\beta}_{\tau}\leq 0),\dots,\tau-P(y_{in_{i}}-x_{in_{i}}^{T}\hat{\beta}_{\tau}\leq 0))^{T}$. We can obtain
\[
\begin{split}
\frac{1}{m}(U(\hat{\beta}_{\tau})-\bar{U}(\hat{\beta}_{\tau})) &=\frac{1}{m}\sum_{i=1}^{m}H_{i}^{T}(\psi_{i}-\varphi_{i})\\
&=\frac{1}{m}\sum_{i=1}^{m}H_{i}^{T}\begin{pmatrix}
P(y_{i1}-x_{i1}^{T}\hat{\beta}_{\tau}\leq 0)-I(y_{i1}-x_{i1}^{T}\hat{\beta}_{\tau}\leq 0)\\
\vdots\\
P(y_{in_{i}}-x_{in_{i}}^{T}\hat{\beta}_{\tau}\leq 0)-I(y_{in_{i}}-x_{in_{i}}^{T}\hat{\beta}_{\tau}\leq 0)
\end{pmatrix}\\
&=\frac{1}{m}\sum_{i=1}^{m}\sum_{j=1}^{n_{i}}h_{ij}[P(y_{ij}-x_{ij}^{T}\hat{\beta}_{\tau}\leq 0)-I(y_{ij}-x_{ij}^{T}\hat{\beta}_{\tau}\leq 0)],
\end{split}
\]
where $h_{ij}$ is a $p\times 1$ vector and $(h_{i1},\dots,h_{in_{i}})=H_{i}^{T}$. According to the uniform strong law of large numbers \citep{Pollard1990}, under condition \ref{itm:A3} we have
\[
\sup_{\hat{\beta}_{\tau}\in \mathfrak{B}}\Biggl\lvert\frac{1}{m}\sum_{i=1}^{m}\sum_{j=1}^{n_{i}}h_{ij}[P(y_{ij}-x_{ij}^{T}\hat{\beta}_{\tau}\leq 0)-I(y_{ij}-x_{ij}^{T}\hat{\beta}_{\tau}\leq 0)]\Biggr\rvert= o(m^{-1/2}) \quad \text{a.s..}
\]
Therefore,
\[
\sup_{\hat{\beta}_{\tau}\in \mathfrak{B}}\lVert\frac{1}{m}(U(\hat{\beta}_{\tau})-\bar{U}(\hat{\beta}_{\tau}))\rVert= o(m^{-1/2}) \quad \text{a.s..}
\]
Now,
\[
G_{m}(\beta_{\tau})=-\frac{1}{m}\frac{\partial\bar{U}(\hat{\beta}_{\tau})}{\partial\hat{\beta}_{\tau}}\Biggr\rvert_{\hat{\beta}_{\tau}=\beta_{\tau}}=\frac{1}{m}\sum_{i=1}^{m}H_{i}^{T}\varGamma_{i}X_{i}
\]
is positive definite and, with probability 1, $G_{m}(\beta_{\tau})\rightarrow G(\beta_{\tau})$ when $m\rightarrow +\infty$. Because $P(y_{ij}-x_{ij}^{T}\beta_{\tau}\leq 0)=\tau$, $\beta_{\tau}$ is the unique solution of the equation $\bar{U}(\hat{\beta}_{\tau})$. Together with $U(\hat{\beta}_{\tau})=0$ and condition \ref{itm:A3}, implies that $\hat{\beta}_{\tau}\rightarrow \beta_{\tau}$ as $m\rightarrow \infty$.

~~ Because $\psi_{i}$ are independent random variables with mean zero, and $\var\{U(\beta_{\tau})/m\}=\frac{1}{m}\sum_{i=1}^{m}X_{i}^{T}\varGamma_{i}\Sigma_{i}^{-1}(\rho)\cov(\psi_{i})\Sigma_{i}^{-1}(\rho)\varGamma_{i}X_{i}$, the multivariate central limit theorem implies that $\frac{1}{\sqrt{m}}U(\beta_{\tau})\rightarrow N(0,V)$.

~~ For any $\hat{\beta}_{\tau}$ satisfying $\lVert \hat{\beta}_{\tau}-\beta_{\tau}\rVert<cm^{-1/3}$,
\[
\begin{split}
U(\hat{\beta}_{\tau})-U(\beta_{\tau})&=\sum_{i=1}^{m}H_{i}^{T}(\hat{\beta}_{\tau})\psi_{i}(\hat{\beta}_{\tau})-\sum_{i=1}^{m}H_{i}^{T}(\beta_{\tau})\psi_{i}(\beta_{\tau})\\
&=\sum_{i=1}^{m}H_{i}^{T}(\hat{\beta}_{\tau})\{\psi_{i}(\hat{\beta}_{\tau})-\psi_{i}(\beta_{\tau})\}+\sum_{i=1}^{m}\{H_{i}^{T}(\hat{\beta}_{\tau})-H_{i}^{T}(\beta_{\tau})\}^{T}\psi_{i}(\beta_{\tau}).
\end{split}
\]
The first term can be written as
\[
\begin{split}
&\sum_{i=1}^{m}H_{i}^{T}(\hat{\beta}_{\tau})\{\psi_{i}(\hat{\beta}_{\tau})-\psi_{i}(\beta_{\tau})\}\\
& \quad =\sum_{i=1}^{m}H_{i}^{T}(\hat{\beta}_{\tau})\varphi_{i}(\hat{\beta}_{\tau})+\sum_{i=1}^{m}H_{i}^{T}(\hat{\beta}_{\tau})\{\psi_{i}(\hat{\beta}_{\tau})-\psi_{i}(\beta_{\tau})-\varphi_{i}(\hat{\beta}_{\tau})\}\\
& \quad =\sum_{i=1}^{m}H_{i}^{T}(\hat{\beta}_{\tau})\varphi_{i}(\hat{\beta}_{\tau})+\sum_{i=1}^{m}H_{i}^{T}(\hat{\beta}_{\tau})\{P(y_{ij}-x_{ij}^{T}\hat{\beta}_{\tau}\leq 0)-I(y_{ij}-x_{ij}^{T}\hat{\beta}_{\tau}\leq 0)\\
& \quad \quad +I(y_{ij}-x_{ij}^{T}\beta_{\tau}\leq 0)-\tau\}
\end{split}
\]
The Lemma in \citet{Jung1996} tells us that
\[
\begin{split}
&\sup\Biggl\lvert \sum_{i=1}^{m}H_{i}^{T}(\hat{\beta}_{\tau})\{P(y_{ij}-x_{ij}^{T}\hat{\beta}_{\tau}\leq 0)-I(y_{ij}-x_{ij}^{T}\hat{\beta}_{\tau}\leq 0)+I(y_{ij}-x_{ij}^{T}\beta_{\tau}\leq 0)-\tau\} \Biggr\rvert\\
& \quad =o_{p}(\sqrt{m}).
\end{split}
\]
Therefore,
\[
\begin{split}
\sum_{i=1}^{m}H_{i}^{T}(\hat{\beta}_{\tau})\{\psi_{i}(\hat{\beta}_{\tau})-\psi_{i}(\beta_{\tau})\}&=\sum_{i=1}^{m}H_{i}^{T}(\hat{\beta}_{\tau})\varphi_{i}(\hat{\beta}_{\tau})+o_{p}(\sqrt{m})\\
&=\bar{U}(\hat{\beta}_{\tau})+o_{p}(\sqrt{m})
\end{split}
\]
From the law of large numbers \citep{Pollard1990} the second term
\[
\begin{split}
\sum_{i=1}^{m}\{H_{i}^{T}(\hat{\beta}_{\tau})-H_{i}^{T}(\beta_{\tau})\}^{T}\psi_{i}(\beta_{\tau})&= \sum_{i=1}^{m}\sum_{j=1}^{n_{i}}(h_{ij}(\hat{\beta}_{\tau})-h_{ij}(\beta_{\tau}))[P(y_{ij}-x_{ij}^{T}\beta_{\tau}\leq 0)\\
& \quad -I(y_{ij}-x_{ij}^{T}\beta_{\tau}\leq 0)]\\
& =o_{p}(\sqrt{m}).
\end{split}
\]
Hence, $U(\hat{\beta}_{\tau})-U(\beta_{\tau})=\bar{U}(\hat{\beta}_{\tau})+o_{p}(\sqrt{m})$. Using Taylor's expansion of $\bar{U}(\hat{\beta}_{\tau})$, we have
\[
\frac{1}{\sqrt{m}}\{U(\hat{\beta}_{\tau})-U(\beta_{\tau})\}=\frac{1}{m}\frac{\partial\bar{U}(\hat{\beta}_{\tau})}{\partial\hat{\beta}_{\tau}}\Biggr\rvert_{\hat{\beta}_{\tau}=\beta_{\tau}}\sqrt{m}(\hat{\beta}_{\tau}-\beta_{\tau})+o_{p}(1).
\]
Because $\hat{\beta}_{\tau}$ is in the $m^{-1/3}$ neighborhood of $\beta_{\tau}$ and $U(\hat{\beta}_{\tau})=0$, we have
\[
\sqrt{m}(\hat{\beta}_{\tau}-\beta_{\tau})=G_{m}^{-1}(\beta_{\tau})\frac{1}{\sqrt{m}}U(\beta_{\tau})+o_{p}(1).
\]
Therefore $\sqrt{m}(\hat{\beta}_{\tau}-\beta_{\tau})\rightarrow N(0,G^{-1}(\beta_{\tau})V\{G^{-1}(\beta_{\tau})\}^{T})$ as $m\rightarrow +\infty$.
\end{proof}

\begin{proof}[Proof of Lemma \ref{lemma}]
Let $\psi_{ij}=\psi_{\tau}(y_{ij}-x_{ij}^{T}\beta_{\tau})$, $\tilde{\psi}_{ij}=\tilde{\psi}_{\tau}(y_{ij}-x_{ij}^{T}\beta_{\tau})$ and $d_{ij}=\varepsilon_{ij}/r_{ij}$, where $\varepsilon_{ij}=y_{ij}-x_{ij}^{T}\beta_{\tau}$, $r_{ij}=\sqrt{x_{ij}^{T}\varOmega x_{ij}}$. Since $\tilde{\psi}_{ij}-\psi_{ij}=\sgn(-d_{ij})\varPhi(-|d_{ij}|)$, where $\sgn(\cdot)$ is the sign function, we have
\[
\begin{split}
\frac{1}{\sqrt{m}}\{\tilde{U}(\beta_{\tau})-U(\beta_{\tau})\}&=\frac{1}{\sqrt{m}}\sum_{i=1}^{m}X_{i}^{T}\varGamma_{i}\Sigma_{i}^{-1}(\rho)
\begin{pmatrix}
\sgn(-d_{i1})\varPhi(-|d_{i1}|)\\
\vdots \\
\sgn(-d_{in_{i}})\varPhi(-|d_{in_{i}}|)
\end{pmatrix}\\
&=\frac{1}{\sqrt{m}}\sum_{i=1}^{m}\sum_{j=1}^{n_{i}}z_{ij}\sgn(-d_{ij})\varPhi(-|d_{ij}|),
\end{split}
\]
where $z_{ij}$ is the $j$th column of $X_{i}^{T}\varGamma_{i}\Sigma_{i}^{-1}(\rho)$. Because
\[
\begin{split}
E(\tilde{\psi}_{ij}-\psi_{ij})&=\int_{-\infty}^{+\infty}\sgn(-d_{ij})\varPhi(-|d_{ij}|)f_{ij}(\varepsilon)d\varepsilon\\
&=\int_{-\infty}^{+\infty}\varPhi(-|\varepsilon|/r_{ij})\{2I(\varepsilon \leq 0)-1\}f_{ij}(\varepsilon)d\varepsilon\\
&=r_{ij}\int_{-\infty}^{+\infty}\varPhi(-|t|)\{2I(t \leq 0)-1\}[f_{ij}(0)+f_{ij}^{'}(\zeta(t))r_{ij}t]dt,
\end{split}
\]
where $\zeta(t)$ is between 0 and $r_{ij}t$. Because $\int_{-\infty}^{+\infty}\varPhi(-|t|)\{2I(t \leq 0)-1\}dt=0$, we have $r_{ij}\int_{-\infty}^{+\infty}\varPhi(-|t|)\{2I(t \leq 0)-1\}f_{ij}(0)dt=0$. Since $\int_{-\infty}^{+\infty}|t|\varPhi(-|t|)dt=1/2$, and by condition \ref{itm:A1}, there exists a constant $M$ such that $\sup_{ij}|f_{ij}^{'}(\zeta(t))|\leq M$. Therefore,
\[
\begin{split}
|E(\tilde{\psi}_{ij}-\psi_{ij})| & \leq r_{ij}^{2}\int_{-\infty}^{+\infty}|t|\varPhi(-|t|)|f_{ij}^{'}(\zeta(t))|dt\\
& \leq Mr_{ij}^{2}/2.
\end{split}
\]
Under regularity conditions \ref{itm:A3} and \ref{itm:A4}, when $m\rightarrow +\infty$,
\[
\biggl\lVert\frac{1}{\sqrt{m}}E\{\tilde{U}(\beta_{\tau})-U(\beta_{\tau})\}\biggr\rVert \leq \frac{1}{\sqrt{m}}\sup_{i,j}|z_{ij}|\sum_{i=1}^{m}Mr_{ij}^{2}/2=o(1).
\]
Moreover,
\[
\frac{1}{m}\var\{\tilde{U}(\beta_{\tau})-U(\beta_{\tau})\}=\frac{1}{m}\sum_{i=1}^{m}\var\biggl\{\sum_{j=1}^{n_{i}}z_{ij}\sgn(-d_{ij})\varPhi(-|d_{ij}|)\biggr\}.
\]
By Cauchy-Schwartz inequality,
\begin{equation*}
\begin{split}
\frac{1}{m}\var\{\tilde{U}(\beta_{\tau})-U(\beta_{\tau})\} & \leq \frac{1}{m}\sum_{i=1}^{m}\sum_{j=1}^{n_{i}}z_{ij}z_{ij}^{T}\var(\tilde{\psi}_{ij}-\psi_{ij})\\
& +\frac{1}{m}\sum_{i=1}^{m}\sum_{j=1}^{n_{i}}\sum_{k\neq j}^{n_{i}}z_{ij}z_{ik}^{T}\sqrt{\var(\tilde{\psi}_{ij}-\psi_{ij})\var(\tilde{\psi}_{ik}-\psi_{ik})}.
\end{split}
\end{equation*}
Hence for each $j=1,\dots,n_{i}$,
\[
\begin{split}
\var(\tilde{\psi}_{ij}-\psi_{ij}) & \leq E(\tilde{\psi}_{ij}-\psi_{ij})^{2}=\int_{-\infty}^{+\infty}\{\sgn(-d_{ij})\varPhi(-|d_{ij}|)\}^{2}f_{ij}(\varepsilon)d\varepsilon \\
&=r_{ij}\int_{-\infty}^{+\infty}\varPhi^{2}(-|t|)f_{ij}(r_{ij}t)dt\\
&=r_{ij}\int_{|t|>\varDelta}^{}\varPhi^{2}(-|t|)f_{ij}(r_{ij}t)dt+r_{ij}\int_{|t|\leq \varDelta}^{}\varPhi^{2}(-|t|)f_{ij}(r_{ij}t)dt\\
&\leq \varPhi^{2}(-\varDelta)+r_{ij}\varDelta f_{ij}(\zeta),
\end{split}
\]
where $\varDelta$ is a positive value, and $\zeta$ is in the interval $(-r_{ij}\varDelta,r_{ij}\varDelta)$. Let $\varDelta=m^{1/3}$. Under condition \ref{itm:A4}, because $r_{ij}=O(m^{-1/2})$, we have $r_{ij}\varDelta=O(m^{-1/6})$. Moreover, both $\varPhi^{2}(-\varDelta)$ and $r_{ij}\varDelta f_{ij}(\zeta)$ converges to 0 as $m\rightarrow +\infty$. By conditions \ref{itm:A2} and \ref{itm:A3}, it can be easily obtained that $\frac{1}{m}\var\{\tilde{U}(\beta_{\tau})-U(\beta_{\tau})\}=o(1)$. Therefore, for any $\beta_{\tau}$, we have $\frac{1}{\sqrt{m}}\{\tilde{U}(\beta_{\tau})-U(\beta_{\tau})\}\rightarrow 0$ as $m\rightarrow +\infty$.
\end{proof}

\begin{proof}[Proof of Theorem \ref{thm2}]
From the results in Theorem \ref{thm1} along with $\sup_{\hat{\beta}_{\tau}\in \mathfrak{B}}\|m^{-1}\{U(\hat{\beta}_{\tau})-\bar{U}(\hat{\beta}_{\tau})\}\|=o(m^{-1/2})$ a.s., and by the triangle inequality, we have $\sup_{\hat{\beta}_{\tau}\in \mathfrak{B}}\|m^{-1}\{\tilde{U}(\hat{\beta}_{\tau})-\bar{U}(\hat{\beta}_{\tau})\}\|=o(m^{-1/2})$. If we denote $\beta_{\tau}$ as the unique solution of equation $\bar{U}(\hat{\beta}_{\tau})=0$ and $\tilde{\beta}_{\tau}$ solving $\tilde{U}(\hat{\beta}_{\tau})=0$, we can obtain that $\tilde{\beta}_{\tau}\rightarrow\beta_{\tau}$ as $m\rightarrow+\infty$.

~~ Before proving the asymptotic normality of $\tilde{\beta}_{\tau}$, we first prove that $m^{-1}\{\tilde{G}(\beta_{\tau})-G(\beta_{\tau})\}\xrightarrow{p} 0$, where $\tilde{G}(\beta_{\tau})=-\partial\tilde{U}(\beta_{\tau})/\partial\beta_{\tau}=\sum_{i=1}^{m}X_{i}^{T}\varGamma_{i}\Sigma_{i}^{-1}(\rho)\tilde{\varLambda}_{i}X_{i}$. If we denote  $H_{i}^{T}=X_{i}^{T}\varGamma_{i}\Sigma_{i}^{-1}(\rho)=(h_{i1},\dots,h_{in_{i}})$, where $h_{ij}$ is a $p\times 1$ vector, we can obtain that
\[
E\{\tilde{G}(\beta_{\tau})\}-G(\beta_{\tau})=\sum_{i=1}^{m}\sum_{j=1}^{n_{i}}h_{ij}\biggl\{\frac{1}{r_{ij}}E\phi\biggl(\frac{\varepsilon_{ij}}{r_{ij}}\biggr)-f_{ij}(0)\biggr\}x_{ij}.
\]
Because
\[
\begin{split}
\biggl|\frac{1}{r_{ij}}E\phi\biggl(\frac{\varepsilon_{ij}}{r_{ij}}\biggr)-f_{ij}(0)\biggr| &=\biggl|\frac{1}{r_{ij}}\int_{-\infty}^{+\infty}\phi\biggl(\frac{\varepsilon}{r_{ij}}\biggr)f_{ij}(\varepsilon)d\varepsilon-f_{ij}(0)\biggr| \\
&=\biggl|\int_{-\infty}^{+\infty}\phi(t)\{f_{ij}(0)+r_{ij}tf_{ij}(\xi_{t})\}dt-f_{ij}(0)\biggr| \\
&=\biggl|r_{ij}\int_{-\infty}^{+\infty}\phi(t)tf_{ij}(\xi_{t})dt\biggr| \\
&\leq r_{ij}\int_{-\infty}^{+\infty}|\phi(t)tf_{ij}(\xi_{t})|dt,
\end{split}
\]
where $\xi_{t}$ lies between 0 and $r_{ij}t$. By condition \ref{itm:A1}, there exists a constant $M$ such that $f_{ij}(\xi_{t})\leq M$. Furthermore, according to condition \ref{itm:A4}, we have
\[
\biggl|\frac{1}{r_{ij}}E\phi\biggl(\frac{\varepsilon_{ij}}{r_{ij}}\biggr)-f_{ij}(0)\biggr|\leq \sqrt{\frac{2}{\pi}}r_{ij}M\rightarrow 0.
\]
By the strong law of large numbers, we know that $m^{-1}\tilde{G}(\beta_{\tau})\rightarrow E\{m^{-1}\tilde{G}(\beta_{\tau})\}$. Using the triangle inequality, we have
\[
|m^{-1}\{\tilde{G}(\beta_{\tau})-G(\beta_{\tau})\}|\leq |m^{-1}\{\tilde{G}(\beta_{\tau})-E\tilde{G}(\beta_{\tau})\}|+|m^{-1}\{E\tilde{G}(\beta_{\tau})-G(\beta_{\tau})\}|\rightarrow o(1),
\]
which is equivalent to $m^{-1}\{\tilde{G}(\beta_{\tau})-G(\beta_{\tau})\}\xrightarrow{p} 0$.

~~ By Taylor series expansion of $\tilde{U}(\hat{\beta}_{\tau})$ around $\beta_{\tau}$ gives us
\[
\tilde{U}(\hat{\beta}_{\tau})=\tilde{U}(\beta_{\tau})-\tilde{G}(\hat{\beta}_{\tau}^{*})(\hat{\beta}_{\tau}-\beta_{\tau}),
\]
where $\hat{\beta}_{\tau}^{*}$ lies between $\hat{\beta}_{\tau}$ and $\beta_{\tau}$. Let $\hat{\beta}_{\tau}=\tilde{\beta}_{\tau}$. Because $\tilde{U}(\tilde{\beta}_{\tau})=0$ and $\tilde{\beta}_{\tau}\rightarrow \beta_{\tau}$, we therefore obtain $\hat{\beta}_{\tau}^{*}\rightarrow \beta_{\tau}$ and $\tilde{G}(\hat{\beta}_{\tau}^{*})\rightarrow \tilde{G}(\beta_{\tau})$. By Lemma \ref{lemma} and $m^{-1}\{\tilde{G}(\beta_{\tau})-G(\beta_{\tau})\}\xrightarrow{p} 0$, we thus have
\[
\sqrt{m}(\tilde{\beta}_{\tau}-\beta_{\tau})=G_{m}^{-1}(\beta_{\tau})\frac{1}{\sqrt{m}}U(\beta_{\tau})+o_{p}(1).
\]
Therefore $\sqrt{m}(\tilde{\beta}_{\tau}-\beta_{\tau})\rightarrow N(0,G^{-1}(\beta_{\tau})V\{G^{-1}(\beta_{\tau})\}^{T})$ as $m\rightarrow +\infty$.
\end{proof}

\end{document}